\documentclass[11pt]{article}%
\usepackage{amsfonts}
\usepackage{amsmath}
\usepackage{fullpage}
\usepackage{hyperref}
\usepackage{amssymb}
\usepackage{graphicx}
\usepackage{color}%
\providecommand{\U}[1]{\protect\rule{.1in}{.1in}}
\newtheorem{theorem}{Theorem}

\newtheorem{definition}[theorem]{Definition}

\newtheorem{remark}[theorem]{Remark}

\newenvironment{proof}[1][Proof]{\noindent\textbf{#1.} }{\ \rule{0.5em}{0.5em}}
\numberwithin{equation}{section}

\begin{document}

\title{Second-order coding rates for pure-loss bosonic channels}
\author{Mark M.~Wilde\thanks{Hearne Institute for Theoretical Physics, Department of
Physics and Astronomy, Center for Computation and Technology, Louisiana State
University, Baton Rouge, Louisiana 70803, USA}
\and Joseph M.~Renes\thanks{Institute for Theoretical Physics, ETH Zurich, 8093
Z\"urich, Switzerland}
\and Saikat Guha\thanks{Quantum Information Processing Group, Raytheon BBN
Technologies, Cambridge, Massachusetts 02138, USA} }
\maketitle

\begin{abstract}
A pure-loss bosonic channel is a simple model for communication over
free-space or fiber-optic links. More generally, phase-insensitive bosonic
channels model other kinds of noise, such as thermalizing or amplifying
processes. Recent work has established the classical capacity of all of these
channels, and furthermore, it is now known that a strong converse theorem
holds for the classical capacity of these channels under a particular
photon-number constraint. The goal of the present paper is to initiate the study of
second-order coding rates for these channels, by beginning with the simplest
one, the pure-loss bosonic channel. In a second-order analysis of
communication, one fixes the tolerable error probability and seeks to
understand the back-off from capacity for a sufficiently large yet finite
number of channel uses. We find a lower bound on the maximum achievable code
size for the pure-loss bosonic channel, in terms of the known expression for
its capacity and a quantity called channel dispersion. We accomplish this by
proving a general \textquotedblleft one-shot\textquotedblright\ coding theorem
for channels with classical inputs and pure-state quantum outputs which reside
in a separable Hilbert space. The theorem leads to an optimal second-order
characterization when the channel output is finite-dimensional, and it remains
an open question to determine whether the characterization is optimal for the
pure-loss bosonic channel.

\end{abstract}

\section{Introduction}

Perhaps the oldest question in quantum information theory was to determine the
classical capacity of optical communication channels with quantum effects
taken into account \cite{gordon1964}. Early work of Holevo in the 1970s
represented progress towards its answer \cite{Holevo73}, but this question
remained unsolved for some time---it was only with the advent of quantum
computation in the 1990s that interest renewed in it. The next major steps
were taken by Hausladen \textit{et al}.~\cite{HJSWW96}\ and then Holevo
\cite{Hol98} and Schumacher and Westmoreland \cite{SW97} (HSW), who
established a general coding theorem for classical communication over quantum
channels. A practically relevant class of channels consists of the
phase-insensitive bosonic channels \cite{WPGCRSL12}, which serve as a model of
optical communication. The first contribution towards understanding
communication over bosonic channels was from Ozawa and Yuen \cite{YO93}. Their
paper established an upper bound on the classical capacity of a noiseless
bosonic channel, which they also showed to be achievable via Shannon's
well-known capacity theorem \cite{S48} and so-called \textquotedblleft
number-state\textquotedblright\ coding. Holevo and Werner then provided a
lower bound on the classical capacity of all phase-insensitive bosonic
channels by considering coherent-state coding ensembles and applying the HSW
theorem to this case \cite{HW01}. Later, Giovannetti \textit{et al}%
.~significantly extended the work of Ozawa and Yuen in \cite{YO93}\ by
completely characterizing the classical capacity of a pure-loss bosonic
channel \cite{GGLMSY04}.

In recent work, a full solution to the original question has now been given.
In particular, Giovannetti \textit{et al}.~have established the classical
capacity of all phase-insensitive bosonic channels~\cite{GHG13}. Furthermore,
Roy Bardhan \textit{et al}.~\cite{BPWW14}, building on prior work in
\cite{BW13,MGH13}, have proven that a strong converse theorem holds for the
classical capacity of these channels under a particular photon-number
constraint, so that the classical capacity of these channels represents a very
sharp dividing line between communication rates that are achievable and unachievable.

Let $M^{\ast}(  \mathcal{P},\varepsilon)  $ denote the maximum
number of messages that can be transmitted over a channel$~\mathcal{P}$ with
failure probability no larger than $\varepsilon\in\left(  0,1\right)  $. The
quantity $\log M^{\ast}(  \mathcal{P},\varepsilon)  $ is known as
the $\varepsilon$-one-shot capacity of $\mathcal{P}$. For $\mathcal{P}$ a
bosonic channel, let $M^{\ast}(  \mathcal{P},N_{S},\varepsilon)  $
denote the same quantity when subject to a photon-number constraint, with
$N_{S}\in\lbrack0,\infty)$.
We can then consider evaluating the quantity
$M^{\ast}(  \mathcal{P}^{\otimes n},N_{S},\varepsilon)  $, i.e.,
when the sender and receiver are allowed $n$ independent uses of the channel
$\mathcal{P}$. When the channel $\mathcal{P}$ is a phase-insensitive bosonic
channel, the lower bound on the classical capacity from coherent-state coding
schemes \cite{HW01}\ combined with the strong converse and the photon-number constraint given in \cite{BPWW14}\ allows
us to conclude that%
\begin{equation}
\lim_{n\rightarrow\infty}\frac{1}{n}\log M^{\ast}\left(  \mathcal{P}^{\otimes
n},N_{S},\varepsilon\right)  =C(  \mathcal{P},N_{S})  ,
\label{eq:1st-order}%
\end{equation}
where $C\left(  \mathcal{P},N_{S}\right)  $ is the classical capacity
(also referred to as Holevo capacity)
of $\mathcal{P}$ with signaling photon number $N_{S}$. Although the statement in
(\ref{eq:1st-order}) is helpful for understanding the information transmission
properties of phase-insensitive bosonic channels (in particular, that the
capacity $C(  \mathcal{P},N_{S})  $ plays the role of a sharp
threshold in the large $n$ limit), it does little to help us understand what
rates are achievable for a given $n$ and fixed error $\varepsilon$, the regime
in which we are interested in practice.

\section{Summary of Results}

The main objective of the present paper is to initiate the study of the
second-order characterization of bosonic channel capacity, in an effort to
understand the newly raised question given above. To do so, we focus on the
pure-loss bosonic channel $\mathcal{N}_{\eta}$, which is a completely positive
trace preserving map resulting from the following Heisenberg evolution:%
\begin{equation}
\hat{b}=\sqrt{\eta}\hat{a}+\sqrt{1-\eta}\hat{e}.\label{eq:pure-loss-trans}%
\end{equation}
In the above, $\eta\in\left[  0,1\right]  $ is the channel transmissivity,
roughly representing the average fraction of photons that make it from sender
to receiver, and $\hat{a}$, $\hat{b}$, and $\hat{e}$ correspond to the
field mode operators for the channel's input, output, and environment, respectively.
Typically, one makes some kind of photon-number constraint on the input of
this channel so that it cannot exceed $N_{S}\in\lbrack0,\infty)$. Without doing so, i.e., letting
$N_{S}=\infty$, the capacity is infinite for any fixed $n$ and $\varepsilon
\neq0$.

The main result of this paper is the following lower bound on $\log M^{\ast
}\left(  \mathcal{N}_{\eta}^{\otimes n},N_{S},\varepsilon\right)  $:%
\begin{equation}
\log M^{\ast}\left(  \mathcal{N}_{\eta}^{\otimes n},N_{S},\varepsilon\right)
\geq ng(  \eta N_{S})  +\sqrt{nv(  \eta N_{S})  }%
\Phi^{-1}(  \varepsilon)  +O(  \log n)  ,
\label{eq:our-formula}%
\end{equation}
whenever $n$ is large enough so that $n\propto1/\varepsilon^{2}$ and there is a mean photon-number constraint. Here $g(
\eta N_{S})  $ is the entropy of a thermal state with mean photon number
$\eta N_{S}$:%
\begin{equation}
g(  x)  \equiv\left(  x+1\right)  \log\left(  x+1\right)  -x\log x.
\end{equation}
Also, $\Phi^{-1}\left(  \varepsilon\right)  $ is the inverse of the cumulative
normal distribution function, so that $\Phi^{-1}\left(  \varepsilon\right)
\leq0$ if and only if $\varepsilon\leq1/2$.

It is known from \cite{GGLMSY04}\ that $g(  \eta N_{S})  $ is equal
to the classical capacity of the pure-loss bosonic channel subject to a mean
photon-number constraint, and that it is
the strong converse capacity when subject to a different photon-number
constraint \cite{WW14}. The quantity $v(  \eta N_{S})  $ is the
\textit{entropy variance} of a thermal state, which we show is equal to%
\begin{equation}
v(  x)  \equiv x\left(  x+1\right)  \left[  \log\left(  x+1\right)
-\log\left(  x\right)  \right]  ^{2}.
\end{equation}
By inspecting (\ref{eq:our-formula}), we can see that the entropy variance
characterizes the back-off from capacity at a fixed error $\varepsilon$ and
for a sufficiently large yet finite $n$.

The above lower bound is reminiscent of the following second-order asymptotic
expansion of $\log M^{\ast}\left(  \mathcal{N}^{\otimes n},\varepsilon\right)
$ when $\mathcal{N}$ is a discrete memoryless classical channel:%
\begin{equation}
\log M^{\ast}\left(  \mathcal{N}^{\otimes n},\varepsilon\right)  =nC\left(
\mathcal{N}\right)  +\sqrt{nV_{\varepsilon}\left(  \mathcal{N}\right)  }%
\Phi^{-1}\left(  \varepsilon\right)  +O\left(  \log n\right)  ,
\label{eq:strassen}%
\end{equation}
where $C\left(  \mathcal{N}\right)  $ is the classical channel capacity and
$V_{\varepsilon}\left(  \mathcal{N}\right)  $ is a channel parameter now known
as the channel dispersion \cite{polyanskiy10}.\footnote{Again, we need $n$
sufficiently large in order for the above equality to hold.} The formula in
(\ref{eq:strassen}) was first identified by Strassen \cite{strassen62} and
later refined by Hayashi \cite{Hay09}\ and Polyanskiy \textit{et
al}.~\cite{polyanskiy10}. See \cite{tan14} for an excellent review of these
developments. In recent work, Tomamichel and Tan have identified that a form
similar to (\ref{eq:strassen}) characterizes $\log M^{\ast}\left(
\mathcal{N}^{\otimes n},\varepsilon\right)  $ for $\mathcal{N}$ a quantum
channel with classical input and a finite-dimensional quantum output
\cite{TT13}. Their main contribution was to establish the inequality $\leq$ in
(\ref{eq:strassen}) for such channels for which the input alphabet is finite
(see \cite{TT13} for details regarding other channels). The inequality $\geq$
in (\ref{eq:strassen}) follows directly from a \textquotedblleft
one-shot\textquotedblright\ coding theorem of Wang and Renner \cite{WR12},
which builds on earlier work of Hayashi and Nagaoka \cite{HN03}, along with an
asymptotic analysis due to Li \cite{li12}\ and Tomamichel and Hayashi
\cite{TH12}. One would like to directly apply these results in order to
recover the bound in (\ref{eq:our-formula}), but cannot do so because the
pure-loss bosonic channel has an infinite-dimensional output. A careful study
of the analysis in \cite{li12} and \cite{TH12}\ reveals that their techniques
are not directly applicable for our setting here.

The rest of this paper proceeds as follows. The next section establishes some
notation and definitions used throughout the rest of the paper.
Section~\ref{sec:one-shot-theorem}\ establishes a one-shot coding theorem for
pure-state classical-quantum channels (those with classical input and a
pure-state classical output). This one-shot coding theorem states that a
quantity known as the $\varepsilon$-spectral inf-entropy \cite{hayashi08}%
\ gives a lower bound on $\log M^{\ast}\left(  W,\varepsilon\right)  $ for any
pure-state classical-quantum channel $W$. Section~\ref{sec:memoryless}\ then demonstrates how
to combine this result with the Berry-Esseen central limit theorem to recover
a second-order lower bound of the form in (\ref{eq:our-formula}) for
pure-state classical-quantum channels. We apply this result to the particular case of a
pure-loss bosonic channel in Section~\ref{sec:bosonic}, recovering the result
stated in (\ref{eq:our-formula}), and we compare the achievable rates with
conventional detection strategies in Section~\ref{sec:comparison}. Finally, we
conclude with a summary and some open questions for future research.

\section{Preliminaries}

\subsection{The $\varepsilon$-spectral inf-entropy}

Let $\rho$ be a density operator, which is a bounded positive semi-definite
operator on a separable Hilbert space, such that its trace is equal to one. We
often use the shorthand $\phi\equiv\left\vert \phi\right\rangle \left\langle
\phi\right\vert $ for a pure state $\left\vert \phi\right\rangle $. We denote
a channel with classical input$~x$ and quantum output $\left\vert \phi
_{x}\right\rangle $ as follows:%
\begin{equation}
W:x\rightarrow\left\vert \phi_{x}\right\rangle , \label{eq:cq-channel}%
\end{equation}
and we refer to it throughout as a pure-state cq (classical-quantum) channel.
Note that the input alphabet can be continuous and the output Hilbert space
can be infinite-dimensional, as is the case for the pure-loss bosonic channel
that we consider in this paper.

Let a spectral decomposition of $\rho$ be given by%
\begin{equation}
\rho=\sum_{z}p_{Z}\left(  z\right)  \left\vert \psi_{z}\right\rangle
\left\langle \psi_{z}\right\vert , \label{eq:spec-decomp-rho}%
\end{equation}
where $\left\{  p_{Z}\left(  z\right)  \right\}  $ is a probability
distribution and $\left\{  \left\vert \psi_{z}\right\rangle \right\}  $ is a
countable orthonormal basis. Let $\gamma\geq0$ and%
\begin{equation}
\Pi_{\gamma}\equiv\left\{  \rho\leq2^{-\gamma}I\right\}  ,
\label{eq:Pi-gamma-def}%
\end{equation}
so that $\Pi_{\gamma}$ projects onto a subspace $\mathcal{T}_{\gamma}$ of the
support of $\rho$ spanned by eigenvectors of $\rho$ with eigenvalues less
than$~2^{-\gamma}$:%
\begin{equation}
\mathcal{T}_{\gamma}\equiv\text{span}\left\{  \left\vert \psi_{z}\right\rangle
:p_{Z}\left(  z\right)  \leq2^{-\gamma}\right\}  . \label{eqn:typ}%
\end{equation}
Let $\underline{H}_{s}^{\varepsilon}\left(  \rho\right)  $ denote the
$\varepsilon$-spectral inf-entropy of $\rho$, defined for $\varepsilon
\in\left[  0,1\right)  $ as \cite{hayashi08}%
\begin{align}
\underline{H}_{s}^{\varepsilon}\left(  \rho\right)   &  \equiv\sup\left\{
\zeta:\text{Tr}\left\{  \rho\Pi_{\zeta}\right\}  \geq1-\varepsilon\right\} \\
&  =\sup\left\{  \zeta:\text{Tr}\left\{  \rho\left(  I-\Pi_{\zeta}\right)
\right\}  \leq\varepsilon\right\}  . \label{eq:eps-spec-inf-def}%
\end{align}
From these definitions, if we set $\gamma
=\underline{H}_{s}^{\varepsilon}\left(  \rho\right)  $, we can conclude that%
\begin{equation}
\text{Tr}\left\{  \Pi_{\gamma}\rho\right\}  \geq1-\varepsilon.
\label{eq:proj-gamma-bound}%
\end{equation}
By employing the spectral decomposition of $\rho$ in (\ref{eq:spec-decomp-rho}%
), we can also express (\ref{eq:eps-spec-inf-def}) as%
\begin{equation}
\underline{H}_{s}^{\varepsilon}\left(  \rho\right)  =\sup\left\{  \zeta
:\Pr_{Z}\left\{  -\log p_{Z}\left(  Z\right)  \leq\zeta\right\}
\leq\varepsilon\right\}  . \label{eq:classical-inf-spectral}%
\end{equation}

\subsection{Central-limit-theorem bounds}

\label{sec:berry-esseen}Let $\Phi\left(  x\right)  $ denote the cumulative
distribution function of a standard normal random variable:%
\begin{equation}
\Phi\left(  x\right)  \equiv\frac{1}{\sqrt{2\pi}}\int_{-\infty}^{x}%
dy\ \exp\left\{  -y^{2}/2\right\}  ,
\end{equation}
and let $\Phi^{-1}\left(  \varepsilon\right)  $ denote its inverse:\ $\Phi
^{-1}\left(  \varepsilon\right)  \equiv\sup\left\{  a\in\mathbb{R}:\Phi\left(
a\right)  \leq\varepsilon\right\}  $ (this is the usual inverse for
$\varepsilon\in\left(  0,1\right)  $ and extends to take values $\pm\infty$
when $\varepsilon$ is outside that interval). The Berry-Esseen theorem gives a
quantitative statement of convergence in the central limit theorem (see, e.g.,
\cite[Section~XVI.5]{F71})\ and plays a prominent role in understanding the
second-order asymptotics of information-processing tasks
\cite{polyanskiy10,TH12,TT13,tan14}.

\begin{theorem}
[Berry-Esseen]Let $X_{1}$, \ldots, $X_{n}$ be a sequence of independent and
identically distributed real-valued random variables, each with mean $\mu$,
variance $\sigma^{2}>0$, and finite third absolute moment, i.e.,
$T\equiv\mathbb{E\{}\left\vert X_{1}\right\vert ^{3}\}<\infty$. Then the
cumulative distribution function of the sum of the standardized versions of
$X_{1}$, \ldots, $X_{n}$\ converges uniformly to that of a standard normal
random variable, with convergence rate $O\left(  1/\sqrt{n}\right)  $. That
is, for all $n\geq1$:%
\begin{equation}
\sup_{x\in\mathbb{R}}\left\vert \Pr\left\{  \frac{1}{\sigma\sqrt{n}}\sum
_{i=1}^{n}\left[  X_{i}-\mu\right]  \leq x\right\}  -\Phi\left(  x\right)
\right\vert \leq\frac{T}{\sigma^{3}\sqrt{n}}. \label{eq:berry-esseen}%
\end{equation}

\end{theorem}

The constant in the upper bound in \eqref{eq:berry-esseen} is due to
\cite{tyurin10}. When evaluated for a tensor-power state$~\rho^{\otimes n}$,
an application of the Berry-Esseen theorem to (\ref{eq:classical-inf-spectral}%
) gives the following asymptotic expansion for$~\underline{H}_{s}%
^{\varepsilon}\left(  \rho^{\otimes n}\right)  $:%
\begin{equation}
\underline{H}_{s}^{\varepsilon}\left(  \rho^{\otimes n}\right)  =nH(
\rho)  +\sqrt{nV(  \rho)  }\Phi^{-1}\left(  \varepsilon
\right)  +O\left(  1\right)  ,
\end{equation}
where the quantum entropy $H(  \rho)  $ and quantum entropy
variance $V(  \rho)  $\ are defined as%
\begin{align}
H(  \rho)   &  \equiv-\text{Tr}\left\{  \rho\log\rho\right\}
=\mathbb{E}_{Z}\left\{  -\log p_{Z}\left(  Z\right)  \right\}
,\label{eq:entropy}\\
V(  \rho)   &  \equiv\text{Tr}\left\{  \rho\left[  -\log
\rho-H\left(  \rho\right)  \right]  ^{2}\right\}  =\text{Var}_{Z}\left\{
-\log p_{Z}\left(  Z\right)  \right\}  . \label{eq:entropy-variance}%
\end{align}
(See, e.g., \cite[Section~IV]{TH12} for more details.)

\subsection{Communication codes for pure-state cq channels}

This section establishes some notation for classical communication over a
pure-state cq channel$~W$, given by (\ref{eq:cq-channel}). The goal is to
transmit classical messages from a sender to a receiver by making use of $W$,
where the messages are labeled by the elements of a set $\mathcal{M}$. Without
loss of generality, any classical communication protocol for $W$ has the
following form: the sender encodes a classical message $m$ into a letter $x$
that is accepted at the input of the channel. Let $\mathcal{E}%
:\mathcal{M\rightarrow X}$\ denote the encoding map and let $x\left(
m\right)  $ denote the codeword corresponding to message~$m$. The sender then
transmits the codeword $x\left(  m\right)  $ over $W$ to the receiver.
Subsequently, the receiver performs a positive operator-valued measure (POVM)
$\{\Lambda_{m}\}_{m\in\mathcal{M}}$ on the system in his possession (i.e., the
operators $\Lambda_{m}$ are positive semi-definite and sum to the identity).
This yields a classical register $\widehat{M}$ which contains his inference
$\hat{m}\in\mathcal{M}$ of the message sent by the sender. The above defines
a \emph{classical code} for the cq channel $W$, which consists of a triple%
\begin{equation}
\mathcal{C}\equiv\left\{  \mathcal{M},\mathcal{E},\{\Lambda_{m}\}_{m\in
\mathcal{M}}\right\}  . \label{eq:code}%
\end{equation}
The size of a code is denoted as $\left\vert \mathcal{C}\right\vert
=|\mathcal{M}|$. The average probability of error for $\mathcal{C}$ on $W$ is%
\begin{equation}
p_{\text{err}}(W,\mathcal{C})\equiv\Pr[\widehat{M}\neq M]=1-\frac
{1}{|\mathcal{M}|}\sum_{m}\text{Tr}\left\{  \Lambda_{m}\phi_{x\left(
m\right)  }\right\}  .
\end{equation}
The following quantity denotes the maximum size of a code for transmitting
classical information over a single use of $W$ with average probability of
error at most~$\varepsilon$.

\begin{definition}
Let $\varepsilon\in(0,1)$ and $W$ be a pure-state cq channel as given in
\eqref{eq:cq-channel}. We define
\begin{equation}
M^{\ast}(W,\varepsilon)\equiv\max\left\{  m\in\mathbb{N}\,:\,\exists
\mathcal{C}:\left\vert \mathcal{C}\right\vert =m\wedge p_{\operatorname{err}%
}(W,\mathcal{C})\leq\varepsilon\right\}  , \label{pr}%
\end{equation}
where $\mathcal{C}$ is a code as prescribed in~\eqref{eq:code}.
\end{definition}

\section{One-shot coding for pure-state cq channels}

\label{sec:one-shot-theorem}Theorem~\ref{thm:main} below establishes a
one-shot lower bound on the maximum achievable code size $M^{\ast}\left(
W,\varepsilon\right)  $ for a pure-state cq channel~$W$ and error
$\varepsilon\in\left(  0,1\right)  $. The main advantage of this theorem over
\cite[Theorem~1]{WR12} is that our lower bound on $\log M^{\ast}\left(
W,\varepsilon\right)  $ is given directly in terms of the $\varepsilon
$-spectral inf-entropy, rather than the hypothesis testing divergence. This in
turn allows us to apply the theorem directly in conjunction with the
Berry-Esseen theorem, in order to establish the lower bound in
(\ref{eq:our-formula}).

\begin{theorem}
\label{thm:main}Let $W$ be a pure-state cq channel as given in
\eqref{eq:cq-channel}, let $p_{X}\left(  x\right)  $ be a probability
distribution over the channel's input alphabet, and let $\rho$ be the expected
density operator at the output:%
\begin{equation}
\rho\equiv\mathbb{E}_{X\sim p_{X}}\left\{  \left\vert \phi_{X}\right\rangle
\left\langle \phi_{X}\right\vert \right\}  . \label{eq:rho-def}%
\end{equation}
Then there exists a codebook for communication over $W$ with average error
probability no larger than $\varepsilon\in(0,1)$, such that the maximum number
of bits $\log M^{\ast}\left(  W,\varepsilon\right)  $\ one can send obeys%
\begin{equation}
\log M^{\ast}\left(  W,\varepsilon\right)  \geq\underline{H}_{s}%
^{\varepsilon-\eta}\left(  \rho\right)  -\log\left(  4\varepsilon/\eta
^{2}\right)  , \label{eq:one-shot}%
\end{equation}
where $\eta\in(0,\varepsilon)$.
\end{theorem}

\begin{proof}
Suppose that $\rho$ has a spectral decomposition as in
(\ref{eq:spec-decomp-rho}), and let $\gamma$ be a parameter such that%
\begin{equation}
\gamma=\underline{H}_{s}^{\varepsilon-\eta}\left(  \rho\right)  .
\end{equation}

We now discuss a coding scheme.

\textbf{Codebook Construction.} Before communication begins, the sender and
receiver agree upon a codebook. We allow them to select a codebook randomly
according to the distribution $p_{X}\left(  x\right)  $. So, for every message
$m\in\mathcal{M}\equiv\left\{  1,\ldots,M\right\}  $, generate a codeword
$x\left(  m\right)  $ randomly and independently according to $p_{X}\left(
x\right)  $.

\textbf{Decoding.} Transmitting the codeword $x\left(  m\right)  $ over the
channel $x\rightarrow\left\vert \phi_{x}\right\rangle $ leads to the state
$\left\vert \phi_{x\left(  m\right)  }\right\rangle $\ at the receiver's end.
Let $\left\vert \phi_{x\left(  m\right)  }^{\gamma}\right\rangle $ be a unit
vector resulting from projecting the codeword $\left\vert \phi_{x\left(
m\right)  }\right\rangle $ onto the subspace given in (\ref{eqn:typ}):%
\begin{equation}
\left\vert \phi_{x\left(  m\right)  }^{\gamma}\right\rangle \equiv\frac
{1}{\left\Vert \Pi_{\gamma}\left\vert \phi_{x\left(  m\right)  }\right\rangle
\right\Vert _{2}}\Pi_{\gamma}\left\vert \phi_{x\left(  m\right)
}\right\rangle . \label{eq:phi-gamma-vector}%
\end{equation}
Upon receiving the quantum codeword $\left\vert \phi_{x\left(  m\right)
}\right\rangle $, the receiver performs a square-root measurement
\cite{B75,B75a} with the following elements, in an attempt to decode the
message $m$:%
\begin{equation}
\Lambda_{m}^{\gamma}\equiv\left(  \sum_{m^{\prime}}\phi_{x\left(  m^{\prime
}\right)  }^{\gamma}\right)  ^{-1/2}\phi_{x\left(  m\right)  }^{\gamma}\left(
\sum_{m^{\prime}}\phi_{x\left(  m^{\prime}\right)  }^{\gamma}\right)  ^{-1/2}.
\label{eq:decoder}%
\end{equation}

\textbf{Error Analysis.} The error when decoding the $m$th codeword is%
\begin{equation}
\text{Tr}\left\{  \left(  I-\Lambda_{m}^{\gamma}\right)  \phi_{x\left(
m\right)  }\right\}  .
\end{equation}
Recall the following operator inequality \cite[Lemma~2]{HN03}%
\begin{equation}
I-\Lambda_{m}^{\gamma}\leq c_{\text{I}}\left(  I-\phi_{x\left(  m\right)
}^{\gamma}\right)  +c_{\text{II}}\sum_{m^{\prime}\neq m}\phi_{x\left(
m^{\prime}\right)  }^{\gamma},\label{eq:HN}%
\end{equation}
where%
\begin{equation}
c_{\text{I}}\equiv1+c,\ \ \ \ \ \ \ \ c_{\text{II}}\equiv2+c+c^{-1},
\end{equation}
and $c$ is a strictly positive number. Applying this operator inequality leads
to the following upper bound on the error:%
\begin{equation}
\text{Tr}\left\{  \left(  I-\Lambda_{m}^{\gamma}\right)  \phi_{x\left(
m\right)  }\right\}  \leq c_{\text{I}}\text{Tr}\left\{  \left(  I-\phi
_{x\left(  m\right)  }^{\gamma}\right)  \phi_{x\left(  m\right)  }\right\}
+c_{\text{II}}\sum_{m^{\prime}\neq m}\text{Tr}\left\{  \phi_{x\left(
m^{\prime}\right)  }^{\gamma}\phi_{x\left(  m\right)  }\right\}  .
\end{equation}
Taking the expectation $\mathbb{E}_{\mathcal{C}}$\ over the random choice of
code gives the following bound:%
\begin{multline}
\mathbb{E}_{\mathcal{C}}\left\{  \text{Tr}\left\{  \left(  I-\Lambda
_{m}^{\gamma}\right)  \phi_{X\left(  m\right)  }\right\}  \right\}  \leq
c_{\text{I}}\ \mathbb{E}_{\mathcal{C}}\left\{  \text{Tr}\left\{  \left(
I-\phi_{X\left(  m\right)  }^{\gamma}\right)  \phi_{X\left(  m\right)
}\right\}  \right\}  \label{eq:HN-error}\\
+c_{\text{II}}\ \mathbb{E}_{\mathcal{C}}\left\{  \sum_{m^{\prime}\neq
m}\text{Tr}\left\{  \phi_{X\left(  m^{\prime}\right)  }^{\gamma}\phi_{X\left(
m\right)  }\right\}  \right\}  .
\end{multline}
By the way that the code is chosen, we have that%
\begin{align}
\mathbb{E}_{\mathcal{C}}\left\{  \text{Tr}\left\{  \left(  I-\phi_{X\left(
m\right)  }^{\gamma}\right)  \phi_{X\left(  m\right)  }\right\}  \right\}   &
=\mathbb{E}_{X\left(  m\right)  }\left\{  \text{Tr}\left\{  \left(
I-\phi_{X\left(  m\right)  }^{\gamma}\right)  \phi_{X\left(  m\right)
}\right\}  \right\}  ,\\
\mathbb{E}_{\mathcal{C}}\left\{  \sum_{m^{\prime}\neq m}\text{Tr}\left\{
\phi_{X\left(  m^{\prime}\right)  }^{\gamma}\phi_{X\left(  m\right)
}\right\}  \right\}   &  =\sum_{m^{\prime}\neq m}\mathbb{E}_{X\left(
m\right)  ,X\left(  m^{\prime}\right)  }\left\{  \text{Tr}\left\{
\phi_{X\left(  m^{\prime}\right)  }^{\gamma}\phi_{X\left(  m\right)
}\right\}  \right\}  .
\end{align}
Consider that%
\begin{align}
\mathbb{E}_{X\left(  m\right)  }\left\{  \text{Tr}\left\{  \phi_{X\left(
m\right)  }^{\gamma}\phi_{X\left(  m\right)  }\right\}  \right\}   &
=\mathbb{E}_{X\left(  m\right)  }\left\{  \frac{1}{\left\langle \phi
_{X^{n}\left(  m\right)  }\right\vert \Pi_{\gamma}\left\vert \phi_{X\left(
m\right)  }\right\rangle }\text{Tr}\left\{  \Pi_{\gamma}\phi_{X\left(
m\right)  }\Pi_{\gamma}\phi_{X\left(  m\right)  }\right\}  \right\}  \\
&  =\mathbb{E}_{X\left(  m\right)  }\left\{  \frac{\left[  \left\langle
\phi_{X\left(  m\right)  }\right\vert \Pi_{\gamma}\left\vert \phi_{X\left(
m\right)  }\right\rangle \right]  ^{2}}{\left\langle \phi_{X\left(  m\right)
}\right\vert \Pi_{\gamma}\left\vert \phi_{X\left(  m\right)  }\right\rangle
}\right\}  \\
&  =\mathbb{E}_{X\left(  m\right)  }\left\{  \left\langle \phi_{X\left(
m\right)  }\right\vert \Pi_{\gamma}\left\vert \phi_{X\left(  m\right)
}\right\rangle \right\}  \\
&  =\mathbb{E}_{X\left(  m\right)  }\left\{  \text{Tr}\left\{  \Pi_{\gamma
}\phi_{X\left(  m\right)  }\right\}  \right\}  \\
&  =\text{Tr}\left\{  \Pi_{\gamma}\rho\right\}  \\
&  \geq1-\left(  \varepsilon-\eta\right)  .
\end{align}
The first equality follows from (\ref{eq:phi-gamma-vector}). The last equality
follows from (\ref{eq:rho-def}) and the last inequality from
(\ref{eq:proj-gamma-bound}). We can then bound the error term in the first
line of (\ref{eq:HN-error}) as follows:%
\begin{equation}
\mathbb{E}_{X\left(  m\right)  }\left\{  \text{Tr}\left\{  \left(
I-\phi_{X\left(  m\right)  }^{\gamma}\right)  \phi_{X\left(  m\right)
}\right\}  \right\}  \leq\varepsilon-\eta.\label{eq:type-i-err-bound}%
\end{equation}
We can bound each of the errors in the second line of (\ref{eq:HN-error}) as%
\begin{align}
&  \mathbb{E}_{X\left(  m\right)  ,X\left(  m^{\prime}\right)  }\left\{
\text{Tr}\left\{  \phi_{X\left(  m^{\prime}\right)  }^{\gamma}\phi_{X\left(
m\right)  }\right\}  \right\}  \nonumber\\
&  =\mathbb{E}_{X\left(  m\right)  ,X\left(  m^{\prime}\right)  }\left\{
\frac{1}{\left\langle \phi_{X\left(  m^{\prime}\right)  }\right\vert
\Pi_{\gamma}\left\vert \phi_{X\left(  m^{\prime}\right)  }\right\rangle
}\text{Tr}\left\{  \Pi_{\gamma}\phi_{X\left(  m^{\prime}\right)  }\Pi_{\gamma
}\phi_{X\left(  m\right)  }\right\}  \right\}  \label{eq:type-ii-bound-1}\\
&  =\mathbb{E}_{X\left(  m^{\prime}\right)  }\left\{  \frac{1}{\left\langle
\phi_{X\left(  m^{\prime}\right)  }\right\vert \Pi_{\gamma}\left\vert
\phi_{X\left(  m^{\prime}\right)  }\right\rangle }\text{Tr}\left\{
\Pi_{\gamma}\phi_{X\left(  m^{\prime}\right)  }\Pi_{\gamma}\mathbb{E}%
_{X\left(  m\right)  }\left\{  \phi_{X\left(  m\right)  }\right\}  \right\}
\right\}  \\
&  =\mathbb{E}_{X\left(  m^{\prime}\right)  }\left\{  \frac{1}{\left\langle
\phi_{X\left(  m^{\prime}\right)  }\right\vert \Pi_{\gamma}\left\vert
\phi_{X\left(  m^{\prime}\right)  }\right\rangle }\text{Tr}\left\{
\phi_{X\left(  m^{\prime}\right)  }\Pi_{\gamma}\rho\Pi_{\gamma}\right\}
\right\}  \\
&  \leq2^{-\gamma}\mathbb{E}_{X\left(  m^{\prime}\right)  }\left\{  \frac
{1}{\left\langle \phi_{X\left(  m^{\prime}\right)  }\right\vert \Pi_{\gamma
}\left\vert \phi_{X\left(  m^{\prime}\right)  }\right\rangle }\text{Tr}%
\left\{  \phi_{X\left(  m^{\prime}\right)  }\Pi_{\gamma}\right\}  \right\}  \\
&  =2^{-\gamma}.\label{eq:type-ii-bound-last}%
\end{align}
The first equality follows from (\ref{eq:phi-gamma-vector}). The second
equality follows from the independence of the codewords corresponding to
different messages. The third equality follows from (\ref{eq:rho-def}) and
cyclicity of the trace. The first inequality is a result of $\Pi_{\gamma}%
\rho\Pi_{\gamma}\leq2^{-\gamma}\Pi_{\gamma}$, which follows from the
definition of $\Pi_{\gamma}$ in (\ref{eq:Pi-gamma-def}). The bounds in
(\ref{eq:type-i-err-bound}) and (\ref{eq:type-ii-bound-1}%
)-(\ref{eq:type-ii-bound-last}) then lead to the following upper bound on the
expectation of the average error probability:%
\begin{equation}
\mathbb{E}_{\mathcal{C}}\left\{  \sum_{m}\text{Tr}\left\{  \left(
I-\Lambda_{m}^{\gamma}\right)  \phi_{X\left(  m\right)  }\right\}  \right\}
\leq c_{\text{I}}\left(  \varepsilon-\eta\right)  +c_{\text{II}}%
\ M\ 2^{-\gamma}.
\end{equation}
So this means there exists a particular codebook with average error
probability less than%
\begin{equation}
c_{\text{I}}\left(  \varepsilon-\eta\right)  +c_{\text{II}}\ M\ 2^{-\gamma}.%
\end{equation}
We would like to have this quantity equal to $\varepsilon$, so we pick $M$ and
$\eta$ such that this is possible:%
\begin{equation}
\varepsilon=c_{\text{I}}\left(  \varepsilon-\eta\right)  +c_{\text{II}%
}\ M\ 2^{-\gamma}.
\end{equation}
We rewrite in terms of $\log M$, so that%
\begin{equation}
\log M=\gamma+\log\left(  \frac{\varepsilon-c_{\text{I}}\left(  \varepsilon
-\eta\right)  }{c_{\text{II}}}\right)  =\underline{H}_{s}^{\varepsilon-\eta
}\left(  \rho\right)  +\log\left(  \frac{\varepsilon-c_{\text{I}}\left(
\varepsilon-\eta\right)  }{c_{\text{II}}}\right)  .
\end{equation}
Choosing $c=\eta/\left(  2\varepsilon-\eta\right)  $ and substituting for
$c_{\text{I}}$ and $c_{\text{II}}$\ then leads to%
\begin{equation}
\log M=\underline{H}_{s}^{\varepsilon-\eta}\left(  \rho\right)  -\log\left(
4\varepsilon/\eta^{2}\right)  ,
\end{equation}
which concludes the proof of the theorem.
\end{proof}

\begin{remark}
The square-root measurement\ in \eqref{eq:decoder} is different from the
original square-root measurement\ constructed in \cite{HJSWW96} because we
take the extra step of normalizing the vectors after projecting them to the
subspace $\mathcal{T}_{\gamma}$. When employing an error analysis along the
lines given above, this normalization appears to be essential in order to
recover the correct second-order asymptotics, at least for the case of cq
pure-state channels with finite-dimensional outputs. That is, the lower bound
in \eqref{eq:2nd-order-expansion} matches the upper bound in \cite{TT13} for
this case.
\end{remark}

\begin{remark}
It would be desirable to establish Theorem~\ref{thm:main}\ when the receiver
employs a sequential quantum decoder
\cite{PhysRevLett.106.250501,PhysRevA.85.012302}. In particular, it would be
desirable if the receiver were able to decode the pure-loss bosonic channel by
performing the \textquotedblleft vacuum-or-not\textquotedblright\ measurement
discussed in \cite{WGTS11}. However, there are two obstacles to be overcome.
First, the standard tool for error analysis of a sequential quantum decoder is
the \textquotedblleft non-commutative union bound\textquotedblright%
\ \cite[Lemma~3]{S11}, but the version of it given in \cite{S11} does not
feature an optimization over a \textquotedblleft$c$\textquotedblright%
\ variable, as is the case with the operator inequality in \eqref{eq:HN}.
Second, the normalization of the vectors in \eqref{eq:phi-gamma-vector}
excludes us from employing the vacuum-or-not measurement. Sen has recently
informed us that it is possible to overcome the first obstacle by modifying
\cite[Lemma~3]{S11} in order to allow for a \textquotedblleft$c$%
\textquotedblright\ variable which can be optimized \cite{S14}. One might be
able to overcome the second obstacle with an error analysis improving upon ours.
\end{remark}

\section{Second-order coding rates for memoryless channels}

\label{sec:memoryless}Of interest in applications is a memoryless pure-state
cq channel $W^{\otimes n}$, defined by%
\begin{equation}
W^{\otimes n}:x_{1}\cdots x_{n}\rightarrow\left\vert \phi_{x_{1}}\right\rangle
\otimes\cdots\otimes\left\vert \phi_{x_{n}}\right\rangle .
\end{equation}
For such a channel, we apply Theorem~\ref{thm:main}, picking codewords
according to an i.i.d.~distribution $\prod_{i=1}^{n} p_{X}(x_{i})$, to find
that it is possible to transmit the following number of bits with average
probability of error no larger than $\varepsilon\in\left(  0,1\right)  $:%
\begin{equation}
\log M^{\ast}\left(  W^{\otimes n},\varepsilon\right)  \geq\underline{H}%
_{s}^{\varepsilon-\eta}\left(  \rho^{\otimes n}\right)  -\log\left(
4\varepsilon/\eta^{2}\right)  . \label{eq:memoryless-LB}%
\end{equation}
By a direct application of the Berry-Esseen theorem as discussed in
Section~\ref{sec:berry-esseen}, choosing $\eta=\frac{1}{\sqrt{n}}$ and $n$
large enough so that $\varepsilon-\eta>0$, we find that the lower bound in
(\ref{eq:memoryless-LB}) expands to%
\begin{equation}
\log M^{\ast}\left(  W^{\otimes n},\varepsilon\right)  \geq nH(
\rho)  +\sqrt{nV(  \rho)  }\Phi^{-1}\left(  \varepsilon
\right)  +O\left(  \log n\right)  , \label{eq:2nd-order-expansion}%
\end{equation}
where $H(  \rho)  $ and $V(  \rho)  $ are defined in
(\ref{eq:entropy})-(\ref{eq:entropy-variance}) and we make use of the fact
that $\Phi^{-1}$ is a continuously differentiable function (see, e.g.,
\cite[Lemma~3.7]{DL14}).

Our derivation applies for pure-state cq channels with outputs in a separable Hilbert space.
However, note that the lower bound in
(\ref{eq:2nd-order-expansion}) also serves as an upper bound as well for channels with
finite-dimensional outputs. This follows because in the finite-dimensional case, we can apply the results of \cite{TT13}, thus establishing the RHS\ of
(\ref{eq:2nd-order-expansion}) as an optimal second-order characterization in
this case.

\section{Application to the Pure-Loss Bosonic Channel}

\label{sec:bosonic}We now apply these results to the case of the pure-loss
bosonic channel $\mathcal{N}_{\eta}$, defined by the transformation in
(\ref{eq:pure-loss-trans}). In this case, we can induce a pure-state cq
channel from the map in (\ref{eq:pure-loss-trans}) by picking a coherent state
$\left\vert \alpha\right\rangle $ \cite{GK04}\ parametrized by $\alpha
\in\mathbb{C}$ and sending the coherent state over the pure-loss bosonic
channel. One key reason why coherent states are good candidates for signaling
over this channel is that the channel retains their purity, only reducing
their amplitude at the output. That is, the output state is $|\sqrt{\eta
}\alpha\rangle$ whenever the input is $\left\vert \alpha\right\rangle $, where
$\eta\in(0,1]$ is the channel transmissivity. So the induced pure-state cq
channel is%
\begin{equation}
\alpha\rightarrow|\sqrt{\eta}\alpha\rangle.
\end{equation}
Since $\eta$ is just a scaling factor, we take $\eta=1$ in what follows for
simplicity. The memoryless version of this channel is simply%
\begin{equation}
\alpha_{1}\cdots\alpha_{n}\rightarrow\left\vert \alpha_{1}\right\rangle
\otimes\cdots\otimes\left\vert \alpha_{n}\right\rangle .
\label{eq:coherent-state-codes}%
\end{equation}
Since all of the codewords are chosen to be coherent states, this leads to
coherent states at the output. The distribution to choose the codewords is an
i.i.d.~extension of the following isotropic, complex Gaussian with
variance$~N_{S}$:%
\begin{equation}
p_{N_{S}}\left(  \alpha\right)  \equiv\frac{1}{\pi N_{S}}\exp\left\{
-\left\vert \alpha\right\vert ^{2}/N_{S}\right\}  .
\end{equation}
The average state of the ensemble is then $n$ copies of a thermal state
$\theta\left(  N_{S}\right)  $:%
\begin{equation}
\theta\left(  N_{S}\right)  \equiv\int d^{2}\alpha\ p_{N_{S}}\left(
\alpha\right)  \left\vert \alpha\right\rangle \left\langle \alpha\right\vert ,
\end{equation}
which is in fact diagonal in the number basis \cite[Sections~2.5\ and~3.8]%
{GK04}, so that%
\begin{equation}
\theta\left(  N_{S}\right)  =\sum_{n=0}^{\infty}\frac{N_{S}^{n}}{\left(
N_{S}+1\right)  ^{n+1}}\left\vert n\right\rangle \left\langle n\right\vert .
\label{eq:thermal-number-basis}%
\end{equation}
The entropy $H\left(  \theta\left(  N_{S}\right)  \right)  $\ of the thermal
state is equal to%
\begin{equation}
H\left(  \theta\left(  N_{S}\right)  \right)  =g(  N_{S})
=\left(  N_{S}+1\right)  \log\left(  N_{S}+1\right)  -N_{S}\log N_{S}.
\end{equation}
As shown in Appendix~\ref{app:disp}, the entropy variance $V\left(
\theta\left(  N_{S}\right)  \right)  $ is given by the expression%
\begin{equation}
V\left(  \theta\left(  N_{S}\right)  \right)  =v(  N_{S})  =
N_{S}\left(  N_{S}+1\right)  \left[  \log\left(  N_{S}+1\right)  -\log\left(
N_{S}\right)  \right]  ^{2}. \label{eq:bosonic-entropy-variance}%
\end{equation}
Thus, from the discussion in the previous section, we can conclude that the
number of bits one can send over $n$ uses of a pure-loss bosonic channel with
failure probability no larger than $\varepsilon\in\left(  0,1\right)  $ has
the following lower bound for sufficiently large yet finite$~n$:%
\begin{equation}
\log M^{\ast}\left(  \mathcal{N}_{\eta=1}^{\otimes n},N_{S},\varepsilon
\right)  \geq ng(  N_{S})  +\sqrt{nv(  N_{S})  }%
\Phi^{-1}\left(  \varepsilon\right)  +O\left(  \log n\right)  .
\end{equation}
Including the channel loss parameter $\eta\in(0,1]$ explicitly leads to the
second-order expansion%
\begin{equation}
\log M^{\ast}\left(  \mathcal{N}_{\eta}^{\otimes n},N_{S},\varepsilon\right)
\geq ng(  \eta N_{S})  +\sqrt{nv(  \eta N_{S})  }%
\Phi^{-1}\left(  \varepsilon\right)  +O\left(  \log n\right)  .
\label{eq:2nd-order}%
\end{equation}

\begin{remark}
The results in \cite{TT13} are not sufficient to recover \eqref{eq:2nd-order}.
The analysis there applies only to channels with finite-dimensional outputs,
as where the pure-loss bosonic channel has an infinite-dimensional output
space. Furthermore, we should clarify that even though our decoder projects
the output space, the projection is onto an infinite-dimensional subspace
because we keep only the photon-number states such that their probabilities
are smaller than a threshold $2^{-\gamma}$ (and this includes photon-number
states with arbitrarily large photon number, yet exponentially small probability).
\end{remark}

\subsection{Photon-number constraint}

\label{sec:pnc}

The development in the previous section ignores imposing a photon-number
constraint, other than choosing the coherent-state codewords according to the
distribution $p_{N_{S}}\left(  \alpha\right)  $. Strictly speaking, we must
impose a photon-number constraint on the codebook, or else the capacity is
infinite. The usual constraint is to impose a mean photon-number constraint
(see, e.g., \cite{HW01,GGLMSY04}). However, as shown in \cite{WW14}, a strong
converse need not hold under such a constraint. Another photon-number
constraint (considered in \cite{WW14}) is to demand that the average codeword
density operator have a large projection onto a photon-number subspace with
photon number no larger than $\left\lceil nN_{S}\right\rceil $, where $n$ is
the number of channel uses and $N_{S}$ is the energy parameter. Specifically,
we might demand that the probability that the average codeword density
operator is outside this subspace decreases exponentially with increasing
blocklength. Under this constraint, the strong converse holds \cite{WW14}. We
can satisfy this demand and decodability for the receiver by choosing coherent
states randomly from an isotropic complex Gaussian distribution with variance
$N_{S}-\delta_{1}$, where $\delta_{1}$ is a small strictly positive real. So,
as long as the number of bits to transmit is of the order%
\begin{equation}
ng\left(  N_{S}-\delta_{1}\right)  +\sqrt{nv\left(  N_{S}-\delta_{1}\right)
}\Phi^{-1}\left(  \varepsilon^{1+\delta_{2}}\right)  +O\left(  \log n\right)
,
\end{equation}
for $\delta_{2}>0$, then we can guarantee that the probability that the
average codeword density operator is outside this subspace decreases
exponentially with $n$. Furthermore, we know that the expectation of the
average error probability is less than $\varepsilon^{1+\delta_{2}}$. We can
then run through the same argument as in \cite{WW14} to conclude
that there exists a code with second-order expansion as given above, such that
we meet the photon-number constraint and the receiver can decode with average
error probability no larger than $\varepsilon$. We just need $n$ large enough
so that%
\begin{equation}
\left[  C\left(  \delta_{1},N_{S}-\delta_{1}\right)  \right]  ^{n/2}%
+\varepsilon^{\delta_{2}}<1,
\end{equation}
where the constant $C\left(  \delta_{1},N_{S}-\delta_{1}\right)  <1$, as
specified in \cite{WW14}.
Thus, at the cost of a degradation in the second-order asymptotics, we can meet the photon-number constraint.
One might be able to circumvent this degradation by
a more advanced approach in which codewords are chosen on the power-limited sphere (see, e.g., \cite{TT15} for details of this approach in the classical case). However, we leave this for future work.

\subsection{Comparison with standard receivers}

\label{sec:comparison}

Heterodyne detection paired with coherent-state encoding is a conventional
strategy for communication over a pure-loss bosonic channel (see, e.g.,
\cite{S09}). By this, we mean that the codewords are of the form in
(\ref{eq:coherent-state-codes}), and the receiver detects every channel output
with a heterodyne receiver. The resulting channel from input to output is
mathematically equivalent to two parallel classical Gaussian channels, so that
the total classical capacity of heterodyne detection is $C_{\mathrm{Het}%
}(N_{S})=\log_{2}(1+N_{S})$ bits/mode, if the codewords have mean photon
number $N_{S}$. In the high photon-number regime, it is well known that the
classical capacity of heterodyne detection asymptotically approaches the
Holevo capacity $g\left(  N_{S}\right)  $ as $N_{S}\rightarrow\infty$
\cite{GGLMSY04}. The classical dispersion of the heterodyne detection receiver
is given by $v_{\mathrm{Het}}(\eta N_{S})$, where
\begin{equation}
v_{\mathrm{Het}}(x)\equiv\frac{x(x+2)\left(  \log_{2}e\right)  ^{2}}%
{(x+1)^{2}}, \label{eq:het-disp}%
\end{equation}
which results from applying the classical dispersion of the additive white noise Gaussian channel (see \cite[Section~IV]{Hay09} or \cite[Theorem~54]%
{polyanskiy10}) to the two parallel channels induced by heterodyne detection.
By comparing (\ref{eq:het-disp}) with (\ref{eq:bosonic-entropy-variance}), it
is straightforward to show that%
\begin{equation}
\lim_{N_{S}\rightarrow\infty}\frac{v(N_{S})}{v_{\mathrm{Het}}(N_{S})}=1.
\end{equation}
This implies that, in the high $N_{S}$ regime, heterodyne detection not only
achieves the capacity, but also the second-order expansion
\eqref{eq:2nd-order-expansion} of the rate as a function of blocklength.

\begin{figure}[ptb]
\centering
\includegraphics[width=4in]{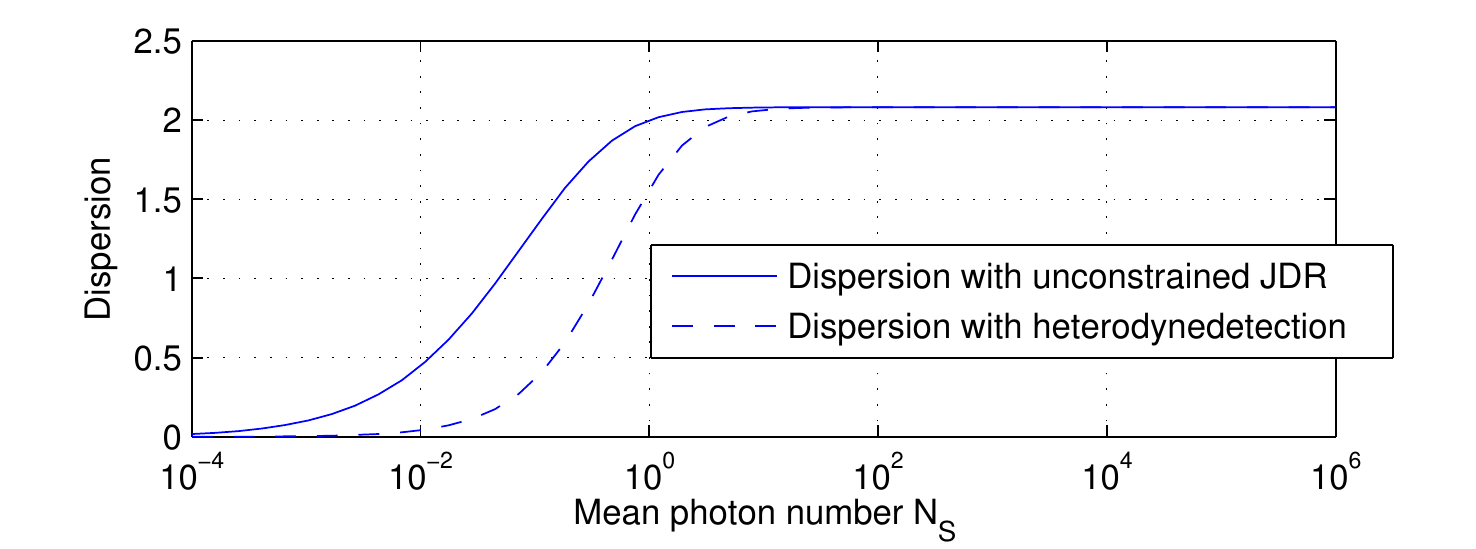}
\caption{Plots of dispersions $v_{\mathrm{Het}}(N_{S})$ (Heterodyne) and
$v(N_{S})$ (optimal JDR) as a function of mean received photon number per mode
$N_{S}$.}%
\label{fig:FiniteBL-3}%
\end{figure}

\begin{figure}[ptb]
\centering
\includegraphics[width=4in]{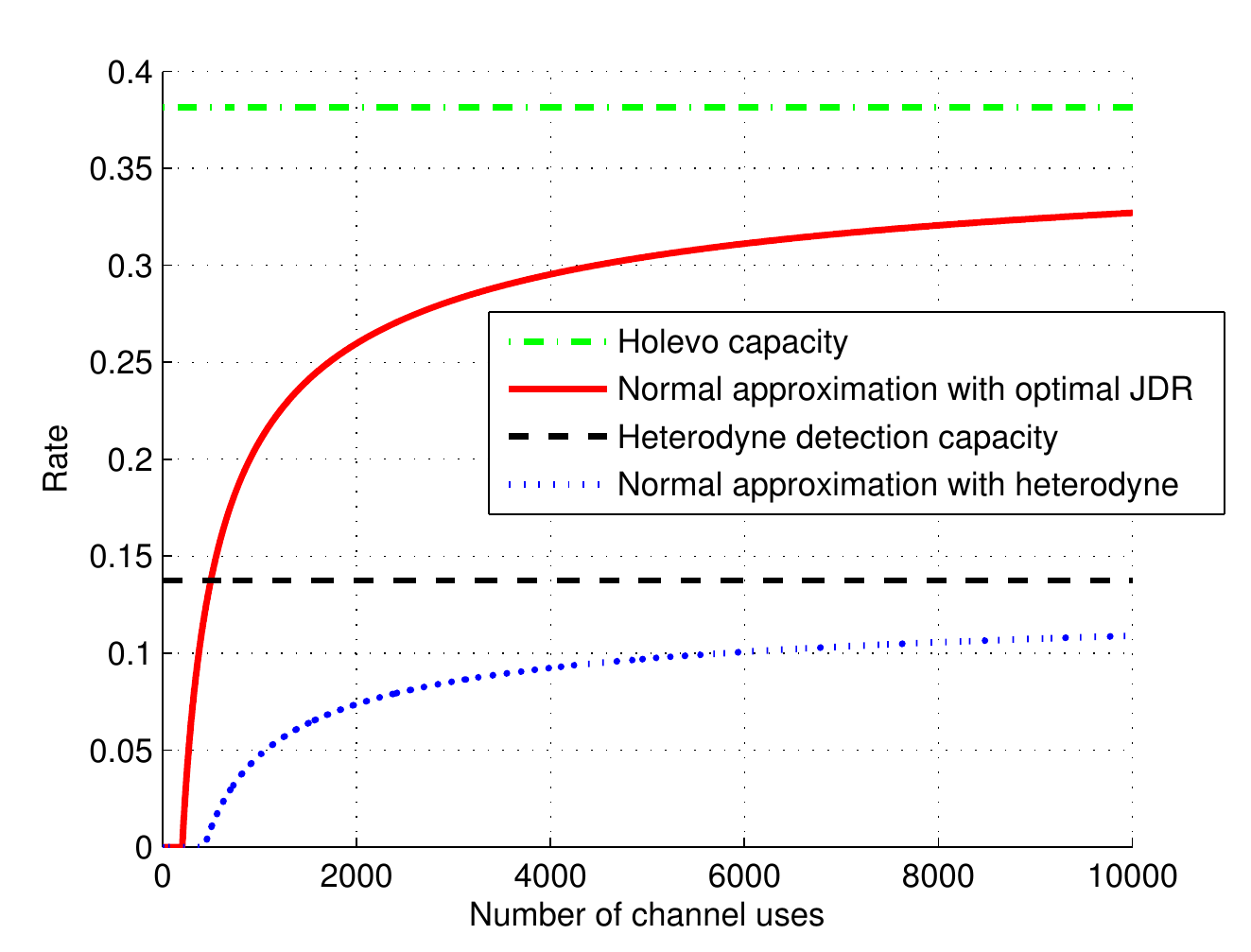} \caption{Normal
approximation for the optimal JDR and the heterodyne detection receiver at a
block error rate threshold of $\varepsilon=10^{-6}$ and with $N_{S} = 0.1$.}%
\label{fig:FiniteBL-2}%
\end{figure}

In Figure~\ref{fig:FiniteBL-3}, we plot $v_{\mathrm{Het}}(N_{S})$ (dashed
line) and $v(N_{S})$ (solid line) as a function of mean received photon number
per mode $N_{S}$. As expected, the most pronounced difference between the
dispersions occurs in the ``deep quantum regime,'' the extremes of which are
for values of $N_{S}$ between $10^{-4}$ and one photon. In
Figure~\ref{fig:FiniteBL-2}, we consider a low mean photon number ($N_{S}%
=0.1$), where there is a gap between the Holevo capacity $g\left(
N_{S}\right)  $\ (blue solid line) and the capacity of heterodyne detection
(red solid line), and we show how the rate for the respective receiver
assumptions increases with the number $n$ of channel uses, at a block error
rate threshold of $\varepsilon=10^{-6}$. However, we 
stress here that this latter plot is intended only to give the reader a rough sense of
which rates are achievable, as they are the ``normal approximation,'' which is accurate only for sufficiently large $n$. 

In the low-photon-number regime, the coherent-state binary phase shift keying
(BPSK) modulation $\left\{  |\alpha\rangle,|-\alpha\rangle\right\}  $, with
$|\alpha|^{2}=N_{S}$ suffices to achieve capacity close to the
unconstrained-modulation Holevo limit. The Holevo limit of the BPSK
constellation is given by%
\begin{equation}
C_{\mathrm{BPSK-Holevo}}({N_{S}})=h_{2}\left(  \frac{1-\langle-\alpha
|\alpha\rangle}{2}\right)  =h_{2}\left(  \frac{1-e^{-2N_{S}}}{2}\right)
\,{\text{bits per mode}},
\end{equation}
where $h_{2}(x)=-x\log x-(1-x)\log(1-x)$ is the binary entropy function. In
units of \textquotedblleft nats\textquotedblright\ per mode, the three
dominant terms in the expansion of $C_{\mathrm{BPSK-Holevo}}({N_{S}})$, and
the unconstrained-modulation Holevo limit $C_{\mathrm{Ultimate-Holevo}}%
(N_{S})=g(N_{S})$, in the $N_{S}\ll1$ regime, are given by%
\begin{align}
C_{\mathrm{Ultimate-Holevo}}(N_{S})  &  \approx-{N_{S}}\ln{N_{S}}+{N_{S}%
}+\frac{N_{S}^{2}}{2}\equiv C_{\infty},\,{\text{and}}\\
C_{\mathrm{BPSK-Holevo}}({N_{S}})  &  \approx-{N_{S}}\ln{N_{S}}+{N_{S}}%
+N_{S}^{2}\ln{N_{S}},
\end{align}
which are identical in the first two terms of the expansion. Note also that
the photon efficiency (nats per photon) goes as $-\ln N_{S}+1$, which
increases without bound as $N_{S}$ decreases towards zero. The maximum
capacity attained by the BPSK modulation, when paired with a receiver that
detects each BPSK symbol one at a time, is attained by the measurement that
discriminates $|\alpha\rangle$ and $\left\vert-\alpha\right\rangle$ with the minimum
probability of error. This minimum error probability is attained by the
Dolinar receiver \cite{D76}, which when used as a receiver for a BPSK
modulation, induces a binary symmetric channel of cross-over probability%
\begin{equation}
\varepsilon=\left[  1-\sqrt{1-|\langle-\alpha|\alpha\rangle|^{2}}\right]  /2=
\left[  1-\sqrt{1-e^{-4N_{S}}}\right]  /2.
\end{equation}
Thus the capacity attained by the best single-symbol measurement
$C_{1}=1-h_{2}(\varepsilon)$, the dominant terms in the expansion of which in
the $N_{S}\ll1$ regime, is given by%
\begin{equation}
C_{\mathrm{BPSK-Dolinar}}({N_{S}})\approx-2N_{S}+o(N_{S})\equiv C_{1}.
\end{equation}
It is clear that the photon efficiency caps off at two nats per photon as
$N_{S}\rightarrow0$. Hence, the gap between $C_{1}$ and $C_{\infty}$ widens as
$N_{S}$ decreases.

\begin{figure}[ptb]
\centering
\includegraphics[width=4.5in]{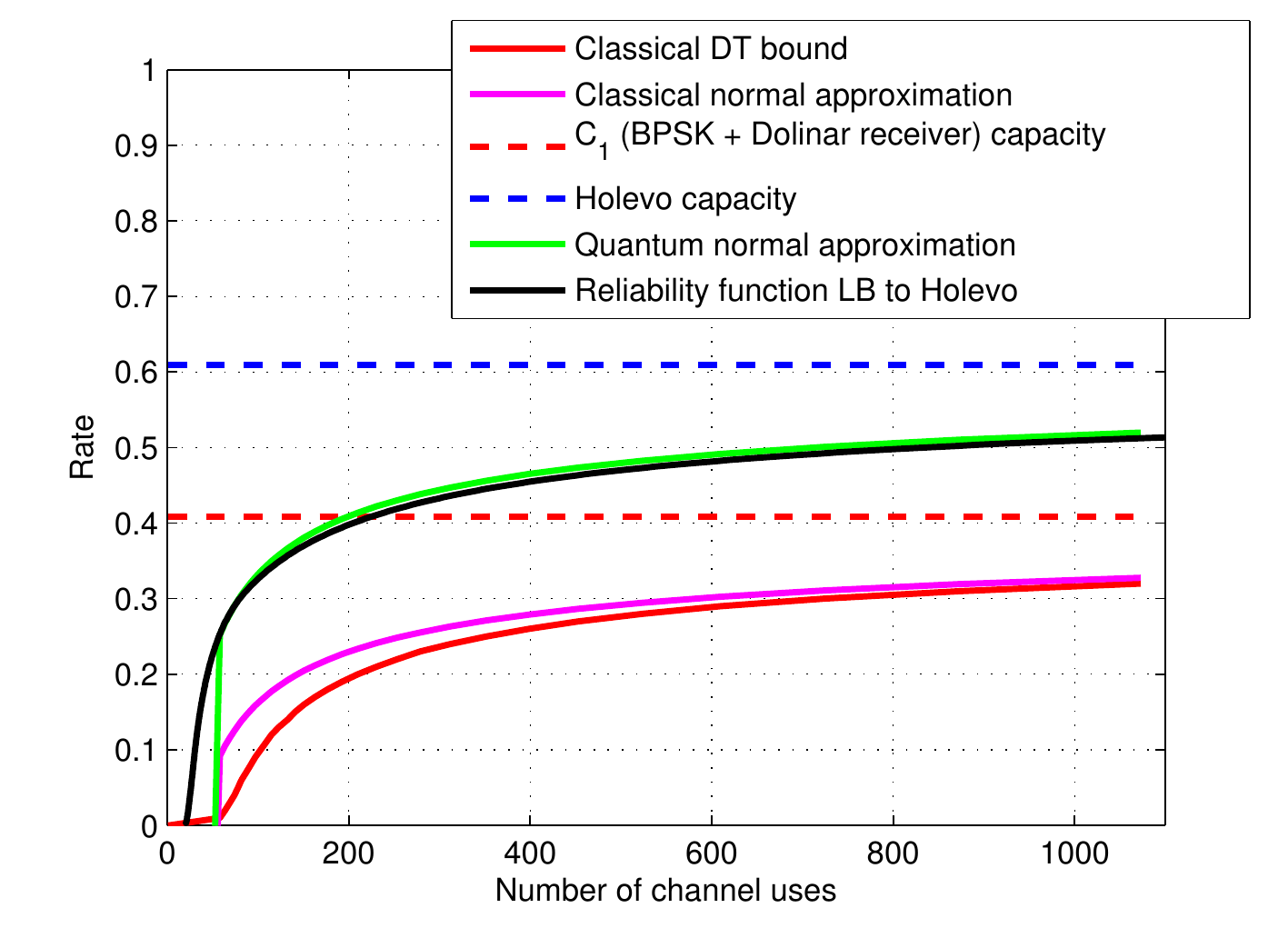} \caption{Rate
achievable by a BPSK alphabet (in bits per BPSK symbol) at a block error rate
threshold of $\varepsilon=10^{-3}$ using the optimal single-symbol receiver
and an optimal JDR at mean photon number $N_{S} = 0.1783$.}%
\label{fig:FiniteBL-1}%
\end{figure}

In Figure~\ref{fig:FiniteBL-1}, we consider a BPSK alphabet of mean photon
number $N_{S}=|\alpha|^{2}=0.1783$ ($e^{-2N_{S}}=0.7$), and plot the optimal
single-symbol receiver capacity $C_{1}$ and the Holevo capacity $C_{\infty}$
(the horizontal red and blue dashed lines). We plot one achievable finite
blocklength rate for BPSK coding, shown by the red solid plot, known as the DT
(dependency testing) bound \cite{polyanskiy10}. The magenta solid line plots
the normal approximation to the second-order rate (capacity minus the
dispersion correction term) as a function of the number $n$ of channel uses.
For the Holevo capacity, we give a finite-blocklength achievability plot, a
bound which derives from the Burnashev-Holevo reliability function (black
line) \cite{BH98}. The green solid line plots the normal approximation to the
second-order rate for the BPSK ensemble as a function of the number $n$ of
channel uses, i.e., evaluated from \eqref{eq:2nd-order-expansion} with $\rho=
(\vert\alpha\rangle\langle\alpha\vert+ \left\vert -\alpha\right\rangle
\left\langle -\alpha\right\vert )/2$. Comparing the reliability function with
$C_{1}$, we see that, for a BPSK constellation with pulses of $\approx0.18$
photons, for a block error rate target of $\varepsilon=10^{-3}$, a
joint-detection receiver will need to act collectively on at least $200$ BPSK
symbols in order for the communication rate to increase beyond what is
possible by the best symbol-by-symbol receiver strategy.

\section{Conclusion}

The main results of this paper are a one-shot coding theorem for arbitrary
pure-state classical-quantum channels and the application of this theorem to
determine second-order coding rates for communication over pure-loss bosonic
channels. The latter result is a first step towards understanding the
second-order asymptotics of communication over bosonic channels.

There are many open questions to consider going forward from here. First and
foremost, it is important to determine whether the formula in
(\ref{eq:our-formula})\ serves as an upper bound for $\log M^{\ast}\left(
\mathcal{N}_{\eta}^{\otimes n},N_{S},\varepsilon\right)  $. At the very least,
one should impose a photon-number constraint similar to that given in
\cite{WW14}---with a mean photon-number constraint, it is already known that
the formula in (\ref{eq:our-formula}) cannot be an upper bound on $\log
M^{\ast}\left(  \mathcal{N}_{\eta}^{\otimes n},N_{S},\varepsilon\right)  $,
due to the lack of a strong converse \cite[Section~2]{WW14}. One could also
consider extending the achievability results developed here to the case of
general phase-insensitive bosonic channels. There is some hope that a complete
understanding of the second-order asymptotics could be developed, now that a
strong converse theorem has been proven in this setting \cite{BPWW14}. Next,
one could also consider evaluating the second-order asymptotics of bosonic
channels when shared entanglement between sender and receiver is available
before communication begins. Some early progress in this direction is in
\cite{DTW14} and references therein. Furthermore, there are many quantum
information-processing tasks for which a first-order characterization is known
(see, e.g., \cite{W11,wildebook13} for a summary), and for which some
second-order characterizations are now known \cite{TH12,TT13,KH13,DL14}. One
could also consider these tasks in the bosonic setting, which could have more
practical applications.

\bigskip

\textbf{Acknowledgements}. We are grateful to Vincent Y.~F.~Tan, Marco
Tomamichel, and Andreas Winter for helpful discussions related to this paper.
The ideas for this research germinated in a research visit of SG\ and MMW\ to
JMR at ETH Zurich in February 2013. SG\ and MMW\ are grateful to Renato
Renner's quantum information group at the Institute of Theoretical Physics of
ETH\ Zurich for hosting them during this visit. MMW\ is grateful for the
hospitality of the Quantum Information Processing Group at Raytheon
BBN\ Technologies for subsequent research visits during August 2013 and April
2014. MMW\ acknowledges startup funds from the Department of Physics and
Astronomy at LSU, support from the NSF\ under Award No.~CCF-1350397, and
support from the DARPA Quiness Program through US Army Research Office award
W31P4Q-12-1-0019. JMR\ acknowledges support from the Swiss
National Science Foundation (through the National Centre of Competence in
Research `Quantum Science and Technology' and grant No.~200020-135048) and the
European Research Council (grant No.~258932). SG was supported by
DARPA's Information in a Photon (InPho)
program, under Contract No.~HR0011-10-C-0159.

\appendix

\section{Calculation of the Bosonic Dispersion}

\label{app:disp} This appendix provides justification for the entropy variance
formula in (\ref{eq:bosonic-entropy-variance}). The key helpful aspect for
calculating the dispersion for the pure-loss bosonic channel is that the
thermal state is diagonal in the number basis, as given in
(\ref{eq:thermal-number-basis}). The eigenvalues in
(\ref{eq:thermal-number-basis}) form a geometric distribution $p\left(
1-p\right)  ^{n}$, where $p\equiv1/\left(  N_{S}+1\right)  $. This
distribution has mean $\left(  1-p\right)  /p=N_{S}$ and variance $\left(
1-p\right)  /p^{2}=N_{S}\left(  N_{S}+1\right)  $, so that the second moment
is $N_{S}\left(  N_{S}+1\right)  +N_{S}^{2}$. With this, we now calculate the
second central moment of the random variable $-\log p_{N}\left(  N\right)  $:%
\begin{align}
v\left(  N_{S}\right)   &  =\sum_{n=0}^{\infty}\frac{1}{N_{S}+1}\left(
\frac{N_{S}}{N_{S}+1}\right)  ^{n}\left\vert -\log\left(  \frac{N_{S}^{n}%
}{\left(  N_{S}+1\right)  ^{n+1}}\right)  -\left[  \left(  N_{S}+1\right)
\log\left(  N_{S}+1\right)  -N_{S}\log N_{S}\right]  \right\vert ^{2}\\
&  =\sum_{n=0}^{\infty}\frac{1}{N_{S}+1}\left(  \frac{N_{S}}{N_{S}+1}\right)
^{n}\left\vert \left(  n-N_{S}\right)  \log\left(  N_{S}+1\right)  -\left(
n-N_{S}\right)  \log\left(  N_{S}\right)  \right\vert ^{2}\\
&  =\sum_{n=0}^{\infty}\frac{1}{N_{S}+1}\left(  \frac{N_{S}}{N_{S}+1}\right)
^{n}\left\vert \left(  n-N_{S}\right)  \left[  \log\left(  N_{S}+1\right)
-\log\left(  N_{S}\right)  \right]  \right\vert ^{2}\\
&  =\sum_{n=0}^{\infty}\frac{1}{N_{S}+1}\left(  \frac{N_{S}}{N_{S}+1}\right)
^{n}\left(  n-N_{S}\right)  ^{2}\left[  \log\left(  N_{S}+1\right)
-\log\left(  N_{S}\right)  \right]  ^{2}\\
&  =\left[  \log\left(  N_{S}+1\right)  -\log\left(  N_{S}\right)  \right]
^{2}\sum_{n=0}^{\infty}\frac{1}{N_{S}+1}\left(  \frac{N_{S}}{N_{S}+1}\right)
^{n}\left(  n^{2}-2nN_{S}+N_{S}^{2}\right)
\end{align}
Using the above facts regarding the geometric distribution, we find that the
sum evaluates to%
\begin{equation}
N_{S}\left(  N_{S}+1\right)  +N_{S}^{2}-2N_{S}^{2}+N_{S}^{2}=N_{S}\left(
N_{S}+1\right)  ,
\end{equation}
so that the variance is equal to%
\begin{equation}
N_{S}\left(  N_{S}+1\right)  \left[  \log\left(  N_{S}+1\right)  -\log\left(
N_{S}\right)  \right]  ^{2}.
\end{equation}

\bibliographystyle{plain}
\bibliography{Ref}

\end{document}